\documentclass[11pt]{article}
\usepackage[english]{babel}
\usepackage[letterpaper,top=2cm,bottom=2cm,left=3cm,right=3cm,marginparwidth=1.75cm]{geometry}

\usepackage[T1]{fontenc}
\usepackage{amsthm,amsmath,amssymb,amsfonts,authblk,setspace}
\usepackage{graphicx}
\usepackage{caption}
\usepackage{subcaption}
\usepackage{xcolor}
\usepackage{authblk}
\usepackage{enumitem}
\usepackage{stfloats}
\usepackage{url}
\usepackage[colorlinks=true, allcolors=blue]{hyperref}
\usepackage{tikz}
\usepackage{pgfplots}
\usepackage{algorithm}
\usepackage{algpseudocode}
\usepackage{float}
\usepackage{textcomp}
\usepackage{listings}
\usepackage{verbatim}
\pgfplotsset{compat=newest}
\usetikzlibrary{decorations.pathreplacing,angles,quotes}

\usepackage{caption}
\usepackage[linesnumbered,algoruled,boxed,lined,algo2e]{algorithm2e}
\usepackage{algorithm}

\newtheorem{theorem}{Theorem}[section]

\newtheorem{lemma}[theorem]{Lemma}

\newtheorem{proposition}[theorem]{Proposition}

\newtheorem{corollary}[theorem]{Corollary}

\newcommand{\E}{\mathop{\mathbb{E}}}

\newcommand{\e}{{\mathbf{e}}}

\newcommand{\SkewTranspose}{{\mathrm{SkewTranspose}}}
\newcommand{\NaiveSkewTranspose}{{\mathrm{NaiveSkewTranspose}}}
\newcommand{\SquareSort}{{\mathrm{SquareSort}}}

\title{SquareSort: a cache-oblivious sorting algorithm}

\author[1]{Michal Kouck{\'{y}}\thanks{Email: koucky@iuuk.mff.cuni.cz. Partially supported by the project of Czech Science Foundation no. 19-27871X and 24-10306S.}}
\author{Josef Matějka\thanks{Email: josef.matejka@tuta.io.}}
\affil[1]{Computer Science Institute of Charles University,\authorcr 
Faculty of Mathematics and Physics,\authorcr
Charles University,\authorcr
Malostransk{\'e}  n{\'a}m\v{e}st\'{\i} 25,\authorcr 
118 00 Praha 1, Czech Republic}

\date{}

\begin{document}
\clearpage
\maketitle

\begin{abstract}
In this paper we consider sorting in the cache-oblivious model of~\cite{frigo2012}. 
We introduce a new simple sorting algorithm in that model which has asymptotically optimal IO complexity $O(\frac{n}{B} \log_{M/B} n)$,
where $n$ is the instance size, $M$ size of the cache and $B$ size of a memory block.
This is the same as the complexity of the best known cache-oblivious sorting algorithm FunnelSort.
\end{abstract}

\thispagestyle{empty}
\newpage
\setcounter{page}{1}

\section{Introduction}

In this paper we consider sorting in the context of cache-oblivious analysis. 
Sorting is perhaps the most fundamental algorithmic problem and we know of dozens of algorithms for it: QuickSort, MergeSort, HeapSort, BubbleSort, ShellSort, \dots
In the context of cache-oblivious algorithms we are aware of only two sorting algorithms: FunnelSort and multiway distribution sort of Frigo et al.~\cite{frigo2012,prokop}
who defined the cache-oblivious model.
Neither of the two algorithms is very natural they seem to be purpose built for the cache-oblivious model.
In cache-oblivious model we analyze the input-output behaviour of algorithms in the external memory model of Aggarwal and Vitter~\cite{AggarwalVitter88} with a cache.
We count the number of block transfers between the cache of size $M$ and the external memory ({\em IO complexity}), where memory blocks are of size $B$.
The two parameters are not known to the algorithm.

In this paper we introduce a new sorting algorithm that we call SquareSort. 
It is randomized and
its expected IO complexity is asymptotically optimal in the cache-oblivious model~\cite{frigo2012,prokop}.
Under the standard {\em tall cache assumption} that $M\ge B^2$, we establish the following bound on its complexity:

\begin{theorem}[Informal]
SquareSort of $n$ items uses $O(\frac{n}{B} \log_{M/B} n)$ IOs in expectation over its randomness.
\end{theorem}

Our algorithm is a natural sorting algorithm. 
It is a variant of distribution sort and one could argue that it is the cache-oblivious analog of QuickSort. 
Its expected running time is $O(n\log n)$ which is asymptotically optimal for comparison-based sorting algorithms.
Our algorithm is similar to the sorting algorithm proposed for various other cache models~\cite{aggarwal1987hierarchical,vitter1994algorithms}
and for parallel models of computation \cite{Reischuk85,ChaudhryCormen05,ChaudhryCormen06,SharmaSen12}.

The algorithm views the array which it should sort as a $\sqrt{n}\times \sqrt{n}$ matrix, it sorts the matrix column-wise,
then it performs a {\em transpose} operation on the matrix and sorts the matrix column-wise once more.
The main new ingredient is the transpose operation which we call {\em $\SkewTranspose$}.
{\em $\SkewTranspose$} is an operation similar to {\em Generalized Matrix Transposition} of Aggarwal et al.~\cite{aggarwal1987hierarchical}.
However according to~\cite[Theorem 3.5]{aggarwal1987hierarchical}, their implementation of the operation in their memory model is losing extra $\log \log n$-factors in its complexity. 
(The paper does not provide details of the implementation and defers the details to the full version of the paper which we could not locate. The paper provides brief reasons for the loss of the two log-log-factors which come down to the implementation.)
Our implementation of $\SkewTranspose$ uses linear time $O(n)$ and the number of IO's is $O(1+n/B)$ which is optimal in our model. 

We provide a full analysis of the algorithm. 
The analysis is not entirely straightforward as one needs to analyze a recurrence relationship with sub-problems of random sizes.
Although the analysis gives what one would have expected and hopped for this was not apriori clear for the following reason.
The recursion is shallow (only $\log \log n$-depth) but very wide. 
At a level of the recursion where instances have size $O(\sqrt{\log n})$ 
the tree already has more than $\Omega(n/\log \log n)$ of those instances so even events that
have {\em exponentially} small probability in the instance size will happen with abundance.
Thus we provide a careful analysis of the algorithm.

We implemented our algorithm and compared its actual running time with the standard C++ library sort and a reference implementation of FunnelSort by R{\o}nn~\cite{ronn}.
Our algorithm is about twice slower than the library sort and twice faster than the FunnelSort.
In Appendix we provide a sample C-code for the algorithm.
Thus our contribution is mostly conceptual---a simple cache-oblivious sorting algorithm.

\subsection{Cache-oblivious analysis}

In this paper we focus on the {\em cache-oblivious analysis} of algorithms \cite{frigo2012,prokop}.
We use the external memory model with a cache of Aggarwal and Vitter~\cite{AggarwalVitter88}. 
We think of the external memory as consisting of cells, each cell can store either a single item (of an array to be sorted), or an integer counter or a pointer.
The external memory can only be accessed via the cache.
The cache has total size $M$ memory cells and it is organized into $M/B$ blocks of size $B$.
The external memory is also partitioned into blocks of $B$ cells.
When accessing some cell in the external memory (during either read or write) the whole block containing that cell is transferred to the cache and made available for processing.
The cache {\em paging algorithm} is responsible for managing the cache and deciding which memory block is stored where in the cache and which block from the cache is evicted if the cache is full and a new memory block needs to be brought into the cache.
(If a block in the cache was modified while in cache it has to be written back to the external memory during its eviction.)

We are interested in the number of block transfers made between the cache and the external memory during an execution of a program. We call the transfers {\em IO's} ({\em input-output operations}).
For a given sequence of memory accesses generated by our program on a specific input, the number of IO's might depend on the paging algorithm.
We make the standard assumption that the paging algorithm is optimal with respect to our program
(and its input) so it generates the least IO's possible for each memory access sequence.
(This assumption is justified for standard paging strategies such as LRU which are optimal up to constant factors.)

Our algorithm is unaware of the actual cache parameters and we analyze it with respect to the parameters $M$ and $B$ in {\em cache-oblivious setting}.
We make the standard {\em tall cache assumption} that $M\ge B^2$.
The general goal is to devise algorithms that give an asymptotically optimal number of IO's regardless of the setting of the two parameters $M$ and $B$.
This has the miraculous effect that such an algorithm is optimal simultaneously for all cache levels in systems with memory hierarchy.
Our sorting algorithm achieves that optimality.

\subsection{Memory management within the external memory.} 

Our algorithm sorts an array with $n$ items (so occupying $n$ continuous memory cells) and outputs the elements into another array of $n$ continuous memory cells.  
Our algorithm uses functions that are called recursively.
We assume that parameters and local variables for each invocation of a function are stored on a single continuous call stack placed somewhere in the external memory.
The call stack is also used to allocate variable size arrays that are local variables within a function.
Arrays are passed to functions as a pointer to the first element of the array.
In particular, if we need to pass a sub-array of an existing array to a function we pass the pointer to the first element of the sub-array.
So passing an array to a function involves $O(1)$ memory accesses.

\section{Our algorithm}

Our sorting algorithm is inspired by the ColumnSort algorithm of Leighton~\cite{leighton}.
It is similar to distribution sort for various other cache models~\cite{aggarwal1987hierarchical,vitter1994algorithms},
sorting in parallel models of computation \cite{Reischuk85,ChaudhryCormen05,ChaudhryCormen06,SharmaSen12}, and it has some similarities to the cache-oblivious distribution sort of Frigo et al.~\cite{frigo1999,prokop}.
Our algorithm sees the array to be sorted as an $m\times m$ matrix stored column-wise,
for $m=\sqrt{n}$. 
The algorithm first recursively sorts each column, then it performs a {\em skew transposition} of the matrix,
and eventually, it sorts again each column recursively. 
The {\em skew transposition} of a matrix and its cache-oblivious implementation is our main new ingredient.
To the best of our knowledge it was not presented before in this form.
Ideally, the skew transposition transposes the matrix so that the first column contains the $m$ smallest elements of the matrix,
the next column contains the next $m$ smallest elements, and so on. 
Thus sorting each column after the transposition sorts the whole matrix.
Figure~\ref{f-sqs} illustrates the algorithm.

\begin{figure}[H]
\begin{tikzpicture}[scale=0.6]
\draw (2,4.5) node{$\sqrt{n}$};
\draw (-1,2) node{$\sqrt{n}$};

\draw (0,0) rectangle +(4,4);
\draw (0.5,0) rectangle +(0.5,4);
\draw (1,0) rectangle +(0.5,4);
\draw (1.5,0) rectangle +(0.5,4);
\draw (2,0) rectangle +(0.5,4);
\draw (2.5,0) rectangle +(0.5,4);
\draw (3,0) rectangle +(0.5,4);
\draw (3.5,0) rectangle +(0.5,4);

\draw [->] (4.5,2) -- (5.5,2);
\draw (5,1.7) node {Sort};

\draw (6.25,2) node[rotate=90] {$>$};
\draw (6.75,2) node[rotate=90] {$>$};
\draw (7.25,2) node[rotate=90] {$>$};
\draw (7.75,2) node[rotate=90] {$>$};
\draw (8.25,2) node[rotate=90] {$>$};
\draw (8.75,2) node[rotate=90] {$>$};
\draw (9.25,2) node[rotate=90] {$>$};
\draw (9.75,2) node[rotate=90] {$>$};

\draw (6,0) rectangle +(4,4);
\draw (6.5,0) rectangle +(0.5,4);
\draw (7,0) rectangle +(0.5,4);
\draw (7.5,0) rectangle +(0.5,4);
\draw (8,0) rectangle +(0.5,4);
\draw (8.5,0) rectangle +(0.5,4);
\draw (9,0) rectangle +(0.5,4);
\draw (9.5,0) rectangle +(0.5,4);

\draw [->] (10.5,2) -- (12.5,2);
\draw (11.5,1.2) node {
	\begin{tabular}{c}
		Skew \\
		transpose \\
	\end{tabular}};
	
\draw (13.5,4.3) node {$<$};
\draw (14,4.3) node {$<$};
\draw (14.5,4.3) node {$<$};
\draw (15,4.3) node {$<$};
\draw (15.5,4.3) node {$<$};
\draw (16,4.3) node {$<$};
\draw (16.5,4.3) node {$<$};

\draw (13,-0.5) rectangle +(0.5,4.5);
\draw (13.5,0.5) rectangle +(0.5,3.5);
\draw (14,-1) rectangle +(0.5,5);
\draw (14.5,1) rectangle +(0.5,3);
\draw (15,0) rectangle +(0.5,4);
\draw (15.5,-0.5) rectangle +(0.5,4.5);
\draw (16,0.5) rectangle +(0.5,3.5);
\draw (16.5,0) rectangle +(0.5,4);

\draw [->] (17.5,2) -- (18.5,2);
\draw (18,1.7) node {Sort};

\draw (19.5,4.3) node {$<$};
\draw (20,4.3) node {$<$};
\draw (20.5,4.3) node {$<$};
\draw (21,4.3) node {$<$};
\draw (21.5,4.3) node {$<$};
\draw (22,4.3) node {$<$};
\draw (22.5,4.3) node {$<$};

\draw (19.25,2) node[rotate=90] {$>$};
\draw (19.75,2) node[rotate=90] {$>$};
\draw (20.25,2) node[rotate=90] {$>$};
\draw (20.75,2) node[rotate=90] {$>$};
\draw (21.25,2) node[rotate=90] {$>$};
\draw (21.75,2) node[rotate=90] {$>$};
\draw (22.25,2) node[rotate=90] {$>$};
\draw (22.75,2) node[rotate=90] {$>$};

\draw (19,-0.5) rectangle +(0.5,4.5);
\draw (19.5,0.5) rectangle +(0.5,3.5);
\draw (20,-1) rectangle +(0.5,5);
\draw (20.5,1) rectangle +(0.5,3);
\draw (21,0) rectangle +(0.5,4);
\draw (21.5,-0.5) rectangle +(0.5,4.5);
\draw (22,0.5) rectangle +(0.5,3.5);
\draw (22.5,0) rectangle +(0.5,4);
\end{tikzpicture}
\caption{An illustration of the SquareSort algorithm.}\label{f-sqs}
\end{figure}
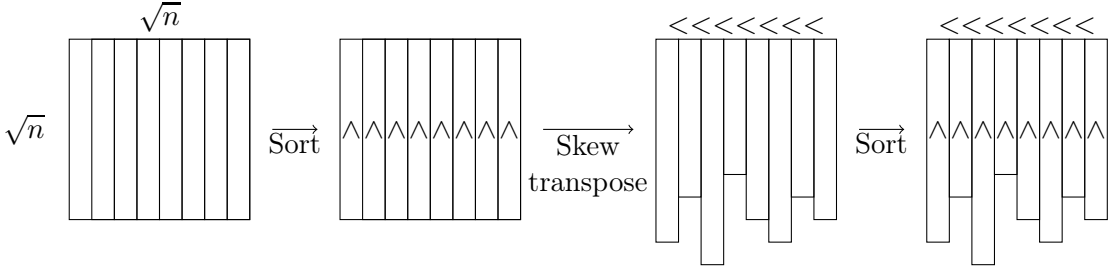

The actual skew transposition uses $m-1$ randomly chosen {\em pivots} $p_1<p_2 < \cdots  <p_{m-1}$, and ``transposes''
the matrix so that the $i$-th column (which we call {\em bucket}) contains elements of the matrix in the range $[p_{i-1},p_i)$.
Here, we set $p_0=-\infty$ and $p_m=\infty$. 
In particular, after the skew transposition, each column of the matrix might have a different size, but in expectation the size
of each bucket is $n/m=O(\sqrt{n})$.

Our key insight is that given any sequence of pivots $p_1,\dots,p_{m-1}$ we can perform the skew transpose operation 
using $O(m+m^2/B)=O(m+n/B)$ input-output operations ({\em IO's}) using the divide-and-conquer strategy similar to usual cache-oblivious matrix transposition.
Indeed, in the special case that each sorted column of the original matrix contains exactly one element from each bucket
the skew transposition coincides with a normal matrix transposition.
Somewhat surprisingly, having the bucket elements distributed unevenly among the columns benefits the IO-complexity of the skew transposition.

The expected number of IO's of the algorithm is governed by the recurrence relationship
\[
    T(n) \le 
\begin{cases}
    O\left(1+\frac{n}{B}\right),                               & \text{if } n\le M/\alpha\\
    \sqrt{n} T(\sqrt{n}) + \E_{n_1+n_2+\dots n_m = n}[ \sum_{i=1}^m T(n_i) ] + O\left(1+\frac{n}{B}\right), & n> M/\alpha
\end{cases}
\]
for some constant $\alpha>0$, which implies $T(n) \le O(\frac{n}{B} \cdot \log_{M/B} n)$.

\subsection{Detailed description of SquareSort}

Here we provide a detailed description of our sorting algorithm. 
The algorithm takes array $A$ as input and sorts it into an array $D$.
For the simplicity of exposition, we assume that all elements in $A$ are distinct.
The order of elements in $A$ might change, and the elements might get permuted. 
Algorithm~\ref{alg-squaresort} gives the pseudo-code of the sorting procedure.

First we partition $A$ into $m$ columns each of size at most $m$, where $m=\lceil \sqrt{n} \rceil$.
For each column, the index of the first element is stored in the array $col[1\dots m]$, and the array $colEnd[1\dots m]$ gives the position of the first element after each column.
Once we determine each column we sort it into the same position in $D$ using a recursive call to $\SquareSort$.
Then we sample $m-1$ distinct pivots other than the minimal element in $D$, and store them sorted in an array $pivots[1\dots m]$ where $pivots[m]=\infty$. 
We give details of this procedure below.
The pivots define buckets of elements, the $i$-th bucket consisting of elements in $D$ from range $[pivots[i-1],pivots[i])$, where $pivots[0]$ is defined to be $-\infty$.
Then we calculate the position of each bucket in the final sorted array. 
This step is a preparation for $\SkewTranspose$ and we will explain its efficient implementation in Section~\ref{sec-skewtranspose}.
Next we call $\SkewTranspose$ which transposes the elements from $D$ back into $A$ so that each bucket forms a continuous part of $A$.
The last step is to recursively sort each bucket from $A$ into $D$.

\begin{algorithm}[H]
\label{algo:squaresort}
   \caption{$\SquareSort(A,D,n)$}\label{alg-squaresort}
   \KwIn{Arrays $A$ and $D$ of size $n$.}
   \KwOut{Sorts items from $A$ into $D$.}
   
   \vspace{1mm}
   \hrule\vspace{1mm}
   \If{$n \leq 16$}{
   
     $\mathrm{simple\_sort}(A, D, n)$; \tcp*[r]{sort small arrays directly}
   
   }
   
   $m = \lceil \sqrt{n} \rceil$;

   Allocate arrays $col[1\dots m]$, $colEnd[1\dots m]$, $pivots[1\dots m]$ and $buc[0\dots m]$;

   \BlankLine

   \For(\tcp*[f]{calculate span of each column}){$i=1,\dots,m$}{$col[i] = 1 + \min((i-1)*m,n)$; $colEnd[i] = 1+\min(i*m,n)$; }

   \BlankLine

   \For(\tcp*[f]{sort each column of $A$ into $D$}){$i=1,\dots,m$}{$\SquareSort(A[col[i], colEnd[i]-1], D[col[i], colEnd[i]-1], colEnd[i]-col[i])$; } 

   \BlankLine

   Sample uniformly at random set $P$ of $m-1$ distinct elements from $D\setminus \{\min(D)\}$ and sort $P \cup \{\infty\}$ into $pivots[1\dots m]$  \tcp*[r]{select pivots}

   \BlankLine

   Calculate $buc[1 \dots m]$, where $buc[j] =  1+|\{ t \in \{1,\dots,n\};\; D[t] < pivots[j]\}|$
   
   \BlankLine

   $\SkewTranspose(D,A,m,col,colEnd,m,pivots,buc)$; \tcp*[r]{skew transpose $D$ into $A$ }

   \BlankLine

   Set $buc[0]=1$  \tcp*[r]{$\SkewTranspose$ shifted $buc[1\dots m]$ by one position}

   \For(\tcp*[f]{sort each bucket from $A$ into $D$}){$j=1,\dots,m$}{$\SquareSort(A[buc[j-1], buc[j]], D[buc[j-1], buc[j]], buc[j]-buc[j-1])$; }  

\end{algorithm}

The selection of distinct pivots from $D$ can be done by sampling a sequence of $m-1$ elements from $D$ uniformly at random with repetition, sorting the sequence, and if there is any repeated element, 
re-sampling the whole set. 
By the converse of the Birthday paradox, in each round of sampling, we succeed with a constant probability to choose distinct elements.
Hence, the expected number of re-sampling is bounded by a constant.

We put the additional requirement that the minimal element of $D$ is not selected as a pivot. 
We can verify this condition by comparing the smallest selected pivot with the smallest element in each sorted column of $D$. 
Again, we re-sample if the condition is violated. 
This happens with a small probability $O(1/\sqrt{n})$.
Since we are selecting $m=O(\sqrt{n})$ pivots the pivot selection is a comparatively cheap operation in terms of IO's.
Thus sorting the pivots can be done by an ordinary MergeSort which uses $O( (m/B) \log m)$ IO's.

We require the pivots to be distinct for the sake of our further analysis.
We provide a working sample C-code in the appendix which deviates slightly from the above description as it selects the pivots independently ignoring repetitions.
However, our C-code can handle repeated elements by creating a single value bucket when multiple pivots are the same and not sorting such single value buckets further.

\subsection{Skew Transposition}\label{sec-skewtranspose}

A skew transposition takes as input an array of $n$ elements partitioned into $m$ sorted columns and a sequence of pivots $p_1<p_2 < \cdots  <p_{m-1}$ defining $m$ buckets, and 
transfers the elements into a destination array so that elements in each bucket form a consecutive part of that array.

To perform skew transposition we first precompute the size of each bucket and determine the final position of each bucket in the destination array.
Computing all the bucket sizes can be done using $O(n/B)$ IO's by reading the matrix column by column and updating the bucket sizes.
This relies on the fact that columns are already sorted.
For each column, we perform a scan of the column simultaneously with the scan of the pivots and update simultaneously the bucket sizes.
This causes $O(n/B)$ IO's.
From the bucket sizes, we can calculate the position of each bucket in the final array by a simple linear scan using $O(m)$ IO's.

Once we calculate the bucket positions we are ready to skew-transpose the matrix.
Let $A$ be the array containing the source matrix, and $D$ be the array prepared for the destination matrix.
$\SkewTranspose$ is a recursive procedure that is invoked with the following parameters:
It gets $\ell$ sub-columns (partial columns) of the source matrix specified by an array of $\ell$ starting indexes of the sub-columns in $A$,
it gets an array of $k$ pivots, and it gets an array of $k$ bucket indexes in $D$ where the elements from the $\ell$ sub-columns shall be stored.
The latter indexes point to the first empty slots in their respective buckets.
(For technical reasons $\SkewTranspose$ also gets an upper bound on the index of each column.)

The transposition moves the elements from the $\ell$ sub-columns in $A$ that are within the range of the requested buckets
into the reserved slots in $D$ and while doing so it increases the indexes of the sub-columns and buckets beyond the used elements.
(We think of all the arrays as being passed by reference.) See Fig.~\ref{f-st} for illustration.

\begin{figure}
\begin{tikzpicture}
\draw (-1.5,4) node {$pivots[1] \le$};
\draw (-1,2.5) node {$col[i]$};
\draw (-1.5,1) node {$< pivots[\ell]$};
\draw (1.25,2.5) node {$i$};
\draw (0,-1) node {$colEnd[i]$};
\draw [<->] (3.2,3.5) -- (3.2,1.5);

\draw [->] (-0.3,4) -- (1.1,3.6);
\draw [->] (-0.3,2.5) -- (1.1,3.4);
\draw [->] (-0.3,2.5) -- (1.1,1.6);
\draw [->] (-0.3,1) -- (1.1,1.4);
\draw [->] (0.8,-1) -- (1,-0.1); 

\draw (0,0) rectangle +(5,5);
\draw (0,1.5) rectangle +(5,2);
\draw (1, 1.5) rectangle +(0.5,2);
\draw (1, 0) rectangle +(0.5,5);
\draw (1.5, 1.5) rectangle +(0.5,2);
\draw (2, 1.5) rectangle +(0.5,2);
\draw (2.5, 1.5) rectangle +(0.5,2);

\draw [->] (5.5,2.5) -- (6.5,2.5);

\draw [->] (3.3,1.6) -- (8.4,-1.1);

\draw (9.75,-1.4) node {$\ell$};
\draw [<->] (8.5,-1.2) -- (11,-1.2); 
\draw (13.3,3) node {$buc[j]$};
\draw (9.25,3) node {$j$};
\draw [->] (12.7,3) -- (9.4,3.9);
\draw [->] (12.7,3) -- (9.4,2.1);

\draw (7,-1) rectangle +(0.5, 6);
\draw (7.5,-0.5) rectangle +(0.5, 5.5);
\draw (8,0) rectangle +(0.5, 5);
\draw (8.5,-1) rectangle +(0.5, 6);
\draw (9,1) rectangle +(0.5, 4);
\draw (9.5,0.5) rectangle +(0.5, 4.5);
\draw (10,-1) rectangle +(0.5, 6);
\draw (10.5,2) rectangle +(0.5, 3);
\draw (11,0) rectangle +(0.5, 5);
\draw (11.5,-1) rectangle +(0.5, 6);
\draw (12,0) rectangle +(0.5, 5);
\draw (8.5,2) rectangle +(0.5,2);
\draw (9,2) rectangle +(0.5,2);
\draw (9.5,2) rectangle +(0.5,2);
\draw (10,2) rectangle +(0.5,2);
\draw (10.5,2) rectangle +(0.5,2);
\end{tikzpicture}
\caption{Illustration of a call to $\SkewTranspose$. Pointers $col[i]$ and $buc[j]$ will advance during the procedure.}\label{f-st}
\end{figure}
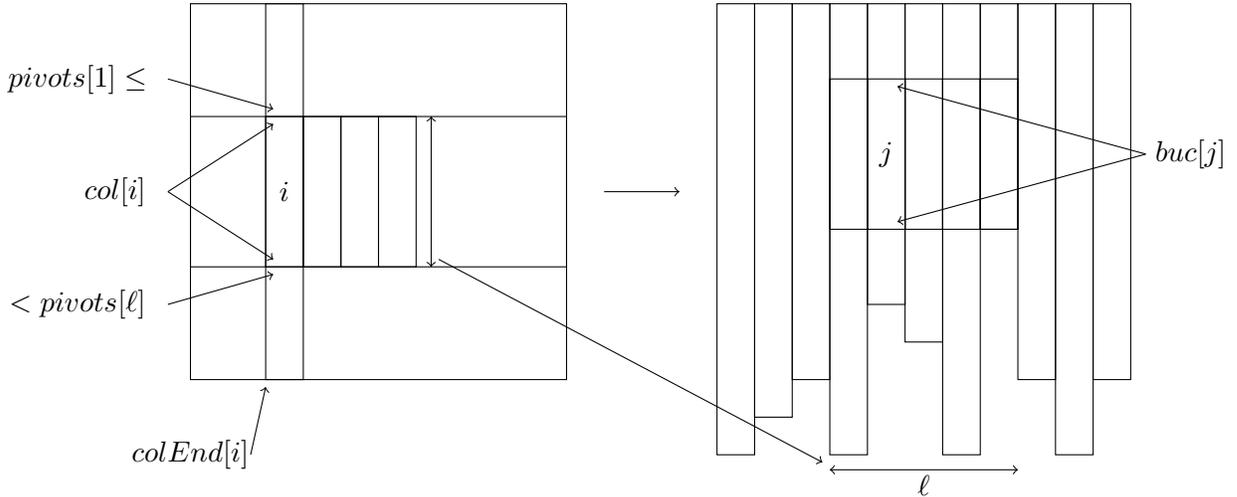

The actual transposition is done by subdividing the problem into four equal-sized sub-problems and solving them recursively:
first we recursively call $\SkewTranspose$ on the first $\ell/2$ sub-columns with the first $k/2$ pivots,
then we call $\SkewTranspose$ on the second $\ell/2$ sub-columns with the first $k/2$ pivots,
then we call $\SkewTranspose$ on the first $\ell/2$ sub-columns with the second $k/2$ pivots,
and finally we call $\SkewTranspose$ on the second $\ell/2$ sub-columns with the second $k/2$ pivots.
Once the dimensions $k$ and $\ell$ reach some small enough threshold we apply a na\"{\i}ve skew transposition algorithm.
This finishes the $\SkewTranspose$. 

Algorithms \ref{alg-skewtranspose} and \ref{alg-naiveskewtranspose} show the pseudo-code of $\SkewTranspose$ and  $\NaiveSkewTranspose$, respectively.

\begin{algorithm}[H]
\label{algo:skewtranspose}
   \caption{$\SkewTranspose(A,D,\ell,col,colEnd,k,pivots,buc)$}\label{alg-skewtranspose}
   \KwIn{A source array $A$, a destination array $D$, starting positions $col[1,\ell]$ of sorted sub-columns in $A$, $colEnd[1,\ell]$ upper-bound positions of sub-columns in $A$, $k$ pivots $pivots[1,k]$ and positions $buc[1,k]$ of free slots in corresponding buckets in $D$. }
   \KwOut{Moves items of sub-columns $col[1,\ell]$ that are less than $pivots[k]$ into their respective buckets in $D$. Updates $col[1,\ell]$ and $buc[1,k]$ which are passed by reference.}
   
   \vspace{1mm}
   \hrule\vspace{1mm}

    \If{$\ell<4$ or $k<4$}{ $\NaiveSkewTranspose(A,D,\ell,col,colEnd,k,pivots,buc)$. }

    $\ell' = \lfloor \ell/2 \rfloor$;

    $k' = \lfloor k/2 \rfloor$;
    
   $\SkewTranspose(A, D, \ell', \hspace{18pt} col[1,\ell'], \hspace{18pt} colEnd[1,\ell'],  \hspace{18pt} k', \hspace{14pt} pivots[1,k'], buc[1,k'])$;

   $\SkewTranspose(A, D, \ell-\ell', col[\ell'+1,\ell], colEnd[\ell'+1,\ell], k', \hspace{14pt} pivots[1,k'], buc[1,k'])$;
   
   $\SkewTranspose(A, D, \ell', \hspace{16pt} col[1,\ell'], \hspace{16pt} colEnd[1,\ell'],  \hspace{16pt} k-k', pivots[k'+1,k], buc[k'+1,k])$;

   $\SkewTranspose(A, D, \ell-\ell', col[\ell'+1,\ell], colEnd[\ell'+1,\ell], k-k', pivots[k'+1,k], buc[k'+1,k])$;

\end{algorithm}

\begin{algorithm}[H]
\label{algo:naiveskewtranspose}
   \caption{$\NaiveSkewTranspose(A,D,\ell,col,colEnd,k,pivots,buc)$}\label{alg-naiveskewtranspose}
   \KwIn{Same parameters as $\SkewTranspose$.}
   \KwOut{Same behavior as $\SkewTranspose$.}
   
   \vspace{1mm}
   \hrule\vspace{1mm}
     
   \For{$i=1,\dots,k$}{
       \For{ $j=1,\dots,\ell$}{
           \While{ $col[j]  < colEnd[j]$ and $A[ col[j] ] \le pivots[i]$ }{

               $D[ buc[i] ] = A[ col[j] ]$;

               $buc[i] = buc[i] + 1$;
               
               $col[j] = col[j] + 1$;
            }
        }
    }
    
\end{algorithm}

\subsubsection{Analysis of $\SkewTranspose$}

In this section we give a cache-oblivious analysis of our $\SkewTranspose$.
Let $n\ge 4$ and $m=\lceil \sqrt{n} \rceil$. 
Let $A$ be the array to be sorted containing $n$ elements and $D$ be the destination array for the sorted items.
We assume that $A$ is partitioned into $m$ columns, each sorted and of size at most $m$.
Let $col[1\dots m]$ and $colEnd[1\dots m]$ be arrays of indexes where for $i=1,\dots,m$, $col[i]$ is the index in $A$ of the first item of the $i$-th column.
Let $pivots[1\dots m]$ be an array containing pivots $p_1\le p_2 \le \cdots \le p_{m-1} \le p_m = \infty$.
Let $buc[1\dots m]$ be an array of indexes into $D$ where $buc[1]=1$ and for $j=2,\dots,m$, $buc[j] = 1+|\{ t \in \{1,\dots,n\};\; A[t] < pivots[j-1]\}|$.

\begin{lemma}
There exists $c_{ST} > 0$ such that for any $n\ge 4$, $m= \lceil\sqrt{n}\rceil$, $B \ge c_{ST}$, $M\ge B^2$, 
$\SkewTranspose(A,B,m,col,m,pivots,buc,colEnd)$ causes at most $c_{ST} (1+n/B)$ IO's.
\end{lemma}

\begin{proof}
Let $c$ be an upper bound on the number of memory cells needed to store parameters and local variables of $\SkewTranspose$ and $\NaiveSkewTranspose$. Let $c_{ST}=\max(515,c)$.
We first analyze the algorithm for $n\ge B^2/8$.

First, we claim that all recursive calls to $\SkewTranspose$ are made with parameters $k$ and $\ell$ differing by at most 1.
This is true for the outermost call where $k=\ell=m$.
Furthermore, if $|k-\ell|\le 1$ then $\left|\left\lceil \frac{k}{2} \right\rceil - \left\lfloor \frac{\ell}{2} \right\rfloor  \right|\le 1$ and $\left|\left\lfloor \frac{k}{2} \right\rfloor - \left\lceil \frac{\ell}{2} \right\rceil \right|\le 1$. 
Since each recursive call of $\SkewTranspose$ is called with $k$ set to either $\left\lceil \frac{k}{2} \right\rceil$ or $\left\lfloor \frac{k}{2} \right\rfloor$ and similarly for $\ell$, 
the property $|k-\ell|\le 1$ is maintained recursively.
In particular, $k$ and $\ell$ reach constant size at about the same time, and the depth of the recursion is at most $\lceil \log m \rceil \le 2+\log n$.

Let $n_{i,j}$ be the number of elements in the $i$-th column of $A$ that belong to the $j$-th bucket that is $n_{i,j}=|\{t \in [col[i],colEnd[i]);\; A[t] \in [pivots[j-1],pivots[j]) \}|$
where we consider $pivots[0]=-\infty$.

Each recursive call to $\SkewTranspose$ or $\NaiveSkewTranspose$ is associated with index sets $I,J \subseteq \{1,\dots,m\}$, $|I|=k, |J|=\ell$,
where $I$ corresponds to the indexes of current $col[1\dots k]$ within the outermost array $col[1\dots m]$, and $J$ corresponds to the indexes of current $buc[1\dots \ell]$ within the outermost array $buc[1\dots m]$.

We claim that if $\SkewTranspose$ is called with parameter $k < \lfloor B/4 \rfloor$ then the number of IO's incurred during processing this call (including all recursive calls)
is at most:
$$
B + \sum_{(i,j)\in I\times J} \frac{2 n_{i,j}}{B}.
$$
We prove the claim first.

Notice, $\ell \le k +1 \le B/4$.
We assume that during processing the call the paging algorithm uses $B/4$ cache blocks to keep the call stack of subsequent recursive calls in the cache.
Since $k<B/4$, the depth of the subsequent recursion is $<B/4$ and by our assumptions local variables and parameters of a single call fit into $c\le B$ memory cells.
Furthermore, the paging algorithm uses $8 \le  B/4$ cache blocks to keep in the cache the parts of the outermost $col[1\dots m]$, $colEnd[1\dots m]$, $buc[1\dots m]$ and $pivots[1\dots m]$
that correspond to indexes $I$ and $J$, respectively.
Additionally, the cache uses $\le B/4$ cache blocks to maintain, for each column of $A$ with an index in $I$, one block in the cache that contains the first item in that column which wasn't transferred to $D$, yet.
Finally, the cache uses $\le B/4$ cache blocks to maintain, for each bucket of $D$ with an index in $J$, one block in the cache that contains the first empty slot in that bucket.
Together this requires $\le B$ cache blocks so it fits into the cache by our tall cache assumption.

At the bottom of the recursion, calls to $\NaiveSkewTranspose$ perform interleaved scans of $k<B/4$ columns of $A$, and transfer elements into $\ell \le B/4$ buckets of $D$.
If a call to $\NaiveSkewTranspose$ transfers $n'$ elements from $A$ to $D$ then it causes $\le 2n'/B$ IO's on $A$ and $D$. 
This is assuming that the first unfinished block of each relevant column of $A$ is already in cache and similarly for the first unfinished blocks of buckets in $D$.
Each $\NaiveSkewTranspose$ continues in scanning $A$ and $D$ from positions where previous $\NaiveSkewTranspose$ calls left.
So in total, all calls to $\NaiveSkewTranspose$ cause at most $2B/4 + \sum_{(i,j)\in I\times J} (2n_{i,j}/B)$ IO's to access $A$ and $D$.
All remaining memory accesses of the $\NaiveSkewTranspose$ calls are either to local variables and parameters or to $col[1\dots k]$, $colEnd[1\dots k]$, $buc[1\dots \ell]$ and $pivots[1\dots \ell]$
which are all kept in the cache. 
Those memory locations require $\le B/2$ memory accesses to be initially loaded into the cache.
We conclude that the $\SkewTranspose$ call with parameter $k < \lfloor B/4 \rfloor$ causes at most: $B + \sum_{(i,j)\in I\times J} (2 n_{i,j}/B)$ IO's.

The outermost call to $\SkewTranspose$ with $k=m$ will cause at most $(m / \lfloor B/8 \rfloor)^2 \le 256 m^2 / B^2$ calls to $\SkewTranspose$ with disjoint products $I\times J$ and $k \in [\lfloor B/8 \rfloor,\lfloor B/4 \rfloor)$.
The tree of the recursive calls with the {\em small} $\SkewTranspose$ calls as leaves has $128 m^2 /B^2 \le 512 n/B^2$ internal nodes.
Each internal node corresponds to a recursive call that causes at most one IO on the call stack (as $c \le B$ so local variables and parameters fit into a single memory block.)
Hence, the total number of IO's caused by invoking the outermost $\SkewTranspose$ is at most:
$$
\frac{512n}{B^2} + \sum_{(I,J)} \sum_{(i,j)\in I \times J} \left(B + \frac{2n_{i,j}}{B}\right) \le \frac{512n}{B^2} + \frac{512 n B}{B^2} + \frac{2n}{B} \le \frac{515n}{B},
$$
The first sum ranges over pairs of $I$ and $J$ corresponding to calls to $\SkewTranspose$ with $k\in [\lfloor B/8 \rfloor,\lfloor B/4 \rfloor)$. 

For $n<B^2/8\le M/8$, the algorithm can store the whole $A$, $D$ and the call stack in the cache so the number of IO's will be bounded by $c_{ST} (1+n/B)$.
\end{proof}

\section{Analysis of SquareSort}

Here we will analyze our SquareSort with respect to the number of IO's in cache-oblivious setting.
As the algorithm is randomized we will analyze its expected cost.
The algorithm is recursive so first we will establish the {\em local} cost of a call to SquareSort
without counting the cost of recursive calls.
We will argue that the expected local cost is $O(1+n/B)$ IO's.
This will establish the following recurrence on the total expected cost $T(n)$:
\begin{eqnarray} \label{rec-IO}
    T(n) \le 
\begin{cases}
    c_1 \left(1+\frac{n}{B}\right),                               & \text{if } n\le M/\alpha\\
    m T(m) + \E_{(n_1,\dots,n_m)\sim \mu_n^m}[ \sum_{i=1}^m T(n_i) ] + c_L\left(1+\frac{n}{B}\right), & n> M/\alpha.
\end{cases}
\end{eqnarray}
for some universal constants $\alpha>1$, and $c_1,c_L >1$, and any $M$ and $B$ where $M \ge B^2 \ge \alpha^2$. 
Here $m=\lceil \sqrt{n} \rceil$ and $\mu_n^m$ denotes the distribution of bucket sizes when distinct pivots are chosen at random.

Finally, we will prove an upper bounds on $T(n)$ establishing our main theorem.

\begin{theorem}\label{t-IO}
There are constants $c\ge 64$, $\beta > \alpha >1$ such that for any $B>\beta$,  $M\ge B^2$, $n\ge 1$
$$
T(n) \le \frac{cn}{B} (2 \max(\log_{M/\alpha} n,1) - 1) + \frac{c}{16}.
$$
\end{theorem}

To analyze the expected number of IO's we will first estimate the worst case space and time complexity of SquareSort.
We start with a claim that the space complexity of the algorithm is $O(n)$ in addition to arrays $A$ and $D$,
so for small values of $n$, all the data can be in the cache simultaneously. We claim:

\begin{proposition}
There is a constant $c_S>1$ such that for any $n\ge 16$, the space used by SquareSort is bounded by $c_S n$ not counting the space used by arrays $A$ and $D$.
\end{proposition}

\begin{proof}
To prove the linear bound on the space complexity we need to bound the depth of the recursion of the SquareSort procedure.
We claim that the depth is at most $n$ and that the maximum size of the call stack is at most $O(n)$ at any moment.
(Here we ignore the space used by $\SkewTranspose$ which we already know is at most linear.)

At each level of the recursion of SquareSort, we use $O(1)$ space to store local variables and parameters
plus we allocate on stack arrays $col, colEnd, pivots, buc$, each of size $m\ge 2$.
Each recursive call is invoked to sort an array of size either $\le m$ (sorting the columns) or $\le n-m+1$ (sorting the buckets).
The latter might lead to a deeper recursion but the size of the sorted array shrinks by at least $m-1$ elements.
So the sorted array shrinks proportionally to the size of allocated arrays on this level of recursion.
Hence in total, the stack corresponding to local variables, parameters and the auxiliary arrays uses space linear in $n$.

Since $\SkewTranspose$ uses also at most linear amount of space on the stack, and similarly the MergeSort, 
the space used by our algorithm is bounded by $c_S n$ for some constant $c_S > 1$.
\end{proof}

Next we bound the worst-case time complexity of SquareSort. We show that it is $O(n^5 \log^2 n)$; a more careful analysis would give bound $o(n^3)$.

\begin{proposition}
There is a constant $c_W>1$ such that for any $n\ge 16$, the worst-case running time of SquareSort is $T_W(n) \le c_W \cdot n^5 \cdot \log^2 n$.
\end{proposition}

\begin{proof}
Consider any run of the algorithm on array of size $n$, and look at the tree of the recursion where each node corresponds to one invocation of SquareSort. 
Consider the subtree of nodes corresponding to calls of SquareSort with arrays $>m=\lceil \sqrt{n} \rceil$. 
The root is part of this subtree.
Each node in the subtree corresponds to local work that takes time at most $c_W n$, for a suitable constant $c_W$.
Each node has at most $m$ children corresponding to sorting columns each of size at most $m$
and at most $m$ children corresponding to sorting buckets.
Each of the latter nodes sorts a bucket with strictly fewer elements than the parent.
So each bucket node has several children sorting smaller buckets which together give the size of the parent bucket.
As each bucket gets subdivided into smaller and smaller buckets the process eventually reaches buckets of size $\le m$.
It is easy to see that the selected subtree has at most $n-1$ nodes corresponding to at most so many bucket subdivisions.
We get the following recurrence for $T_W(n)$:
\begin{eqnarray*}
    T_W(n) &\le& (n-1) \cdot m \cdot T_W(m) + n \cdot T_W(m) + c_W n^2 \\
    &\le& n\cdot (m+1) \cdot T_W(m) + c_W n^2 \\
    &\le& c_W \sum_{i=1}^{2\log \log n} \left( 2^{i-1} \cdot n^{\frac{1}{2^{i-1}}} \right)^2 \cdot \prod_{j=1}^{i-1} 2^j n^{\frac{3}{2^j}},
\end{eqnarray*}
where the last inequality follows by iterating the recurrence and using $m+1 \le 2\sqrt n$, for $n\ge 4$.
Notice,  $\frac{1}{2^{i-2}} + \sum_{j=1}^{i-1} \frac{3}{2^j} \le 2.5$ for any $i\ge 1$, hence,
\begin{eqnarray*}
T_W(n) &\le& c_W \cdot 2^{2\log \log n} \cdot 2^{(2\log \log n)^2/2} \cdot n^{2.5} \\
       &\le& c_W \log^2 n \cdot n^2 \cdot n^{2.5} \le c_W \cdot n^5 \log^2 n.
\end{eqnarray*}
Here we used $(\log \log n)^2 \le \log n$, for $n\ge 16$.
\end{proof}

\subsection{Establishing recurrence~(\ref{rec-IO})}

We are ready to prove recurrence~(\ref{rec-IO}). 
We set $\alpha = 2(c_S + 2)$ and $c_1=6c_S$. 
If $n \le M/\alpha$ then all the memory used by SquareSort can be stored in the cache simultaneously.
Hence, in the case that $n\le  M/\alpha$, we can bound the number of IO's by $c_1(1+n/B)$ as required.

So now we focus on the case $n > M/\alpha$ and we will count the local cost of a SquareSort call.

Assuming $B>16$,  $\mathrm{simple\_sort}(A, D, n)$ for $n<16$ implemented by e.g. InsertionSort
will cause at most 5 IO's (at most 2 IO's on $A$, at most 2 IO's on $D$ and 1 IO for local variables).

Initializing $col[1\dots m]$ and $colEnd[1\dots m]$ will cause at most $2 (m/B) + 3$ IO's.

Sampling uniformly at random $m-1$ candidate pivots from $D$ will cause $m \le 4(c_S + 2) \frac{n}{B}$ IO's.
(Indeed, if $n\ge M/\alpha \ge B^2 / \alpha$ then $\sqrt{n} \ge B/(2(c_S+2))$ 
so $mB \le 2\sqrt{n} B \cdot \frac{2(c_S+2)}{2(c_S+2)} \le 4(c_S+2) n$, hence $m \le 4(c_S+2) \cdot \frac{n}{B}$.)
The sampling is done by selecting a random index $t\in \{1,\dots,n\}$ and picking $D[t]$ as a candidate pivot.
After selecting the pivots and storing them in $pivots[1\dots m-1]$ (which costs $\le 1+m/B$ IO's)
we sort $pivots[1\dots m-1]$ using MergeSort.

MergeSort has IO complexity bounded by $c_M (1+\frac{m}{B} \log m) \le c_M (1+\frac{n}{B})$, for some constant $c_M > 1$.
(Here we assume $n>16$ so $m \log m \le n$.)
Checking that the sorted pivots are distinct and that the smallest one is not the minimal element of $D$
costs one scan over the pivots and one scan over $D$ so at most $2\frac{n}{B} + 2$ IO's.

If the check fails we have to try the whole sampling again. 
The probability that we select $m-1$ distinct pivots larger than $\min(D)$ is at least 
$\prod_{i=1}^{m-1} (1-\frac{i-1}{n-1}) \ge (1-\frac{m-2}{n-1})^{m-1} \ge \e^{-2 (m-2)(m-1)/(n-1)} \ge 1/\e^2$.
(Here we used the facts: $(1-x) \ge \e^{-2x}$ for $x\le 1/2$, $(m-2)/(n-1) \le 1/2$, and $(m-2)(m-1)/(n-1) <1$ for $n>2$.)
So the expected number of repetitions before we succeed sampling distinct pivots is $\e^2$.
Thus the expected cost of sampling the pivots is $\le \e^2 [4(c_S+2) + c_M + 2]\cdot \frac{n}{B} + \e^2 [c_M + 3]$.

As explained in Section~\ref{sec-skewtranspose}, calculating $buc[1\dots m]$ costs $\le c_B(1+n/B)$ for some suitable constant $c_B \ge 1$, and $\SkewTranspose$ costs $\le c_{ST} (1+n/B)$ IO's.
Preparing the parameters for $2m$ recursive calls to SquareSort can cause $O(m)$ IO's in total.
Hence, the total local cost is $\le c_L(1+\frac{n}{B})$ IO's, for some universal constant $c_L > 1$.
Hence, the expected cost of SquareSort satisfies recurrence~(\ref{rec-IO}).

\subsection{Analysis of expected bucket sizes}

To prove the main theorem we need to establish some useful properties of the distribution on bucket sizes.
Let $\mu_n^m$ denote a distribution on vectors $(n_1,\dots,n_m) \in \{1,\dots,n-1\}^n$ which is obtained by sampling uniformly at random a set of elements $p_0<p_1<\cdots < p_{m-1} \in \{1,\dots, n\}$ and setting $n_i = p_i - p_{i-1}$, for $i<m$, and $n_m=p_0+n-p_{m-1}$.
It is easy to see that for each vector $(n_1,\dots,n_m)$ in the support of $\mu_n^m$, $\sum_{i=1}^m n_i = n$.
Moreover, the marginal distribution of each $n_i$ is the same and we denote it by $\mu_n^{m|1}$. 
We can always shift all the pivots by subtracting $p_0-1$ without affecting $n_1,n_2,\dots, n_m$ so 
the distribution $\mu_n^m$ corresponds to the distribution of bucket sizes when pivots $p_1,\dots,p_{m-1}$ are selected at random from among elements of rank $\{2,\dots,m\}$ and $p_0$ is set to be the element of rank 1. 

Now we establish the probability of large deviation of each $n_i$ from its expectation.

\begin{proposition}\label{p-distrupper}
For any $n\ge 100$, $m=\lceil \sqrt{n}\rceil$, $i\in \{1,\dots, m\}$, $t\ge 1$:
$\Pr_{(n_1,\dots,n_m)\sim \mu_n^m}[n_i \ge t\sqrt{n}] \le  \e^{-0.9 t + 0.1}.$
\end{proposition}

\begin{proof}
Since each $n_i$ is distributed as $n_1$, we can focus our attention on $n_1$ and assume $p_0=1$. $n_1 \ge t\sqrt{n}$ if no element of rank $\le \lceil t\sqrt{n} \rceil$ is selected as a pivot $p_1,\dots,p_{m-1}$. 
We need to select the $m-1$ distinct pivots $p_1,\dots, p_{m-1}$ uniformly at random from elements of rank $\{2,\dots,n\}$. 
We can select the set of pivots by drawing $m-1$ uniformly random elements from the set of elements of rank $\{2,\dots,n\}$ one by one (with replacement),
and then re-sampling elements that are equal to another element selected earlier. 
We repeat the re-sampling until we obtain a set of $m-1$ distinct pivots with ranks from $\{2,\dots, n\}$.
Clearly this will give a uniformly random set of $m-1$ pivots.
In order, for $n_1 \ge t\sqrt{n}$, the first $m-1$ sampled elements must not contain any element of rank $\{2,\dots,\lceil t\sqrt{n} \rceil\}$. Hence
$$
\Pr_{n_1\sim \mu_n^{m|1}}\left[n_1 \ge t\sqrt{n}\right] \le  \left(1-\frac{t\sqrt{n}-1}{n}\right)^{\sqrt{n}-1} \le \e^{-\frac{(t \sqrt{n} - 1)\cdot (\sqrt{n} -1) }{n}} \le \e^{-t\cdot \frac{ n - \sqrt{n} }{n} + \frac{\sqrt{n} - 1 }{n}} \le  \e^{- 0.9 t + 0.1}.
$$
\end{proof}

We will need also an estimate on the deviation of each $n_i$ in the other direction.

\begin{proposition}\label{p-supper}
For any $n\ge 100$, $m=\lceil \sqrt{n}\rceil$, $s\ge 2$:
$\Pr_{(n_1,\dots,n_m)\sim \mu_n^m}[n_1 \le s] \le  \e^{2} s/\sqrt{n}.$
\end{proposition}

\begin{proof}
Consider the process of generating $n_1,n_2,\dots,n_m$ by selecting each pivot $p_1,p_2,\dots,p_{m-1}$ uniformly and independently at random from $\{1,\dots,n\}$
and re-sampling all of the pivots if they are not all distinct or some of them is the minimal element. 
Eventually, for each $i=1,\dots,m$, set $n_i=p_i-p_{i-1}$, where $p_0=1$ and $p_m=n+1$.
Clearly, if $n_1 \le s$ then at least one of the pivots during the last round was selected from the range $\{2,\dots,s\}$.
We can upper bound the probability of this event by the expected number of pivots selected from that range during any of the iterations.
In expectation there will be at most $\e^2$ iterations as each iteration succeeds with probability at least $1/ \e^2$ (which was observed earlier).
So we will sample at most $\e^2 (m-1)$ pivots in expectation.
A given pivot is sampled from $\{2,\dots,s\}$ with probability at most $\le 2/n$, 
so the expected number of pivots sampled in $\{2,\dots,s\}$ is at most $\e^2 (m-1) s /n \le \e^2 s / \sqrt{n}$.
The claim follows.
\end{proof}

The following claim is the main technical lemma that allows us to deal with expectation over bucket sizes.

\begin{lemma}
    For any integer $n\ge 100$ and $m=\lceil \sqrt{n}\rceil$:
    $$
    \E_{(n_1,\dots,n_m)\sim \mu_n^m}\left[ \sum_{i=1}^m n_i \log n_i \right] \le \frac{1}{2} n \log n + 4\e n.
    $$
\end{lemma}

\begin{proof}
First, we claim that $\E_{n_1 \sim \mu_n^{m|1} } [ n_1 \log \lceil \frac{n_1}{\sqrt{n}} \rceil] \le 3 \e \sqrt{n}$.
To see this, we use $\Pr[ t\sqrt{n} \le n_1 <(t+1)\sqrt{n} ] \le \e^{-0.9 t+0.1}$ implied by the previous proposition and
we group possible sizes of $n_1$ as follows
\begin{eqnarray*}
\E_{n_1 \sim \mu_n^{m|1} } \left[ n_1 \log \left\lceil \frac{n_1}{\sqrt{n}} \right\rceil \right] &=& \sum_{\ell=1}^{n-1} \Pr[n_1=\ell] \cdot \ell \cdot \log \left\lceil \frac{\ell}{\sqrt{n}} \right\rceil \\
&\le& \sum_{t\ge 0} \Pr\left[ t\sqrt{n} \le n_1 <(t+1)\sqrt{n} \right] \cdot (t+1) \cdot \sqrt{n} \cdot \log (t+1) \\
&\le& \sqrt{n} \cdot \sum_{t\ge 0} \e^{-0.9t+0.1} \cdot (t+1) \cdot  \log (t+1) \\
&\le& \e \sqrt{n} \cdot \sum_{t\ge 0} \e^{-0.9(t+1)} \cdot (t+1)^2 \\
&\le& 3\e\sqrt{n}.
\end{eqnarray*}
Now, using the linearity of expectation
\begin{eqnarray*}
\E_{(n_1,\dots,n_m)\sim \mu_n^m}\left[ \sum_{i=1}^m n_i \log n_i \right] 
&=& \E_{(n_1,\dots,n_m)\sim \mu_n^m}\left[ \sum_{i=1}^m n_i \log \left(n_i \cdot \frac{\sqrt{n}}{\sqrt{n}} \right) \right] \\
&\le& \E_{(n_1,\dots,n_m)\sim \mu_n^m}\left[ \sum_{i=1}^m n_i \log \sqrt{n} + \sum_{i=1}^m n_i \log \left\lceil \frac{n_i}{\sqrt{n}} \right\rceil \right] \\
&=& \E_{(n_1,\dots,n_m)\sim \mu_n^m}\left[ \sum_{i=1}^m n_i \log \sqrt{n} \right] + \sum_{i=1}^m \E_{(n_1,\dots,n_m)\sim \mu_n^m}\left[n_i \log \left\lceil \frac{n_i}{\sqrt{n}} \right\rceil \right] \\
&\le& n \log \sqrt{n} + \sum_{i=1}^m 3\e\sqrt{n} \\
&\le& n \log \sqrt{n} + 3\e n + 3\e\sqrt{n} \le  n \log \sqrt{n} + 4\e n.
\end{eqnarray*}

\end{proof}

We derive from the lemma the following corollary.

\begin{corollary}\label{c-expectation}
    For any integer $n\ge 100$, $m=\lceil \sqrt{n}\rceil$ and $s\ge 2$:
    $$
    \E_{(n_1,\dots,n_m)\sim \mu_n^m}\left[ \sum_{i=1}^m n_i \max(\log_s n_i,1) \right] \le \frac{1}{2} n \log_s n + \frac{4\e n}{\log s} + 2\e^2s^2.
    $$
\end{corollary}

\begin{proof}
\begin{eqnarray*}
\E_{(n_1,\dots,n_m)\sim \mu_n^m}\left[ \sum_{i=1}^m n_i \max(\log_s n_i,1) \right] &\le& \E_{(n_1,\dots,n_m)\sim \mu_n^m}\left[ \sum_{i=1}^m n_i \log_s n_i \right] 
+ \sum_{i=1}^m \Pr_{n_i \sim \mu_n^{m|1}} [n_i \le s] \cdot s \\
&\le& \frac{1}{\log s} \cdot \E_{(n_1,\dots,n_m)\sim \mu_n^m}\left[ \sum_{i=1}^m n_i \log n_i \right] + m\cdot  \frac{\e^2 s}{\sqrt{n}} \cdot s \\
&\le& \frac{1}{2} n \log_s n + \frac{4\e n}{\log s} + 2\e^2 s^2.   
\end{eqnarray*}
Here the second inequality follows from Proposition~\ref{p-supper} and the last one from the previous lemma using the linearity of expectation.
\end{proof}

\subsection{Proof of the main theorem}

We will prove Theorem~\ref{t-IO} by induction on $n$. The base case will be covered by the following lemma.
Recall $\alpha$ from recurrence~(\ref{rec-IO}).

\begin{lemma}\label{lem-basis}
    There are constants $c_{\mathrm{IB}}>1$ and $\gamma \ge \alpha$ such that for any $B>\gamma$, $M\ge B^2$, $1 \le n \le (M/\alpha)^3$:
$$
T(n) \le c_{\mathrm{IB}} \left(\frac{n}{B} + 1\right).
$$
\end{lemma}

\begin{proof}
 Let $r\ge 25$ be such that $135 \log^{3/2} r \le r^{1/4}$.
 Hence, $\e \le r$ and $r^{3/4}+1 \le \frac{10}{9} r^{3/4}$.
 Let $\gamma = \max(r,\alpha)$.
 Let  $s=M/\alpha$.
 Clearly, $s \ge B^2 / \alpha \ge \gamma^2 / \alpha \ge \gamma \ge r$.
 For $n\le s$, the claim follows from the recurrence~(\ref{rec-IO}).
 For $s < n \le s^3$ we claim that with probability at least $1-1/s^{17}$, 
 the recursion tree of calls to SquareSort has depth at most two before all calls are on arrays of size $\le s$.
 Indeed, the probability that any bucket produced by the first call is of size larger than $25 s^{3/2} \log s \ge 22.5 \lceil s^{3/2} \rceil \log s$
 is at most $s^3 \e^{0.1}/s^{20.25} \le 1/s^{17}$ by Proposition~\ref{p-distrupper}.
 If all the buckets happen to be small then on the second level of the recursion tree 
 we have at most $2\lceil \sqrt{n} \rceil \le 4 s^{3/2}$ nodes each corresponding to sorting an array of size at most $25 s^{3/2} \log s$.
 Each of the nodes generates smaller buckets, altogether at most $8s^3$ buckets.
 Since $\lceil \sqrt{25 s^{3/2} \log s} \rceil \le 6 s^{3/4} \log^{1/2} s$,
 the probability that any of the smaller buckets is of size larger than $135 s^{3/4} \log^{3/2} s \ge  22.5 \log s \cdot 6 s^{3/4} \log^{1/2} s \ge 22.5 \log s \cdot \lceil \sqrt{25 s^{3/2} \log s} \rceil $
 is at most $8\e^{0.1} s^3 / s^{20.25} \le 8/s^{17}$. 
 Given that $s\ge r$, $135 s^{3/4} \log^{3/2} s \le s$.
 So the recursion tree reaches nodes sorting arrays of size $\le s$ in at most two rounds unless a bad event happens with probability $9/s^{17}$.

 The sum of the sizes of nodes on the second level is $2n$ and on the third level it is $4n$ so $\le 7n$ in total including the first layer. 
 (The total size doubles at each level.)
 Each subproblem of size $\ge s$ generates at most $2\lceil \sqrt{s} \rceil \le 4 \sqrt{s}$ smaller subproblems.
 Thus, the number of subproblems of size $<s$ that are generated by a subproblem of size $\ge s$ is at most $\frac{7n}{s} \cdot 4 \sqrt{s} \le 28 \frac{n}{\sqrt{s}} \le  28 n\sqrt{\alpha}/B = O(n/B)$.
 (The last inequality follows from $s \ge B^2/\alpha$.)
 Each subproblem of size $\ell <s$ is processed using at most $O(1+\ell/B)$ IO's.
 So if the bad event does not happen we will perform in total $O(1+n/B)$ IO's.
 With probability $\le 9/s^{17}$ we might get larger buckets on the third level.
 The worst-case number of IOs is at most $c_W n^5 \log^2 n \le c_W s^{17}$, so the expected contribution to IO complexity
 in the case of bad event happening is constant.
 The lemma follows.
\end{proof}

We are ready to prove the upper bound on IO complexity of SquareSort.

\begin{proof}[Proof of Theorem~\ref{t-IO}]
Let $c_{\mathrm{IB}}$ and $\alpha \le \gamma$ be as in Lemma~\ref{lem-basis}.
Set $\beta = \max(2^{8(5+8\e+4\e^2)}, \gamma)$.
Set $c=\max(32 c_{\mathrm{IB}}, 16 c_L)$. 
Assume $B>\beta$,  $M\ge B^2$, $n\ge 1$ are given.
Let $s=M/\alpha$. 
Hence, $B>\gamma>\alpha$ and $s\ge B$.
For $n\le s^3$ the conclusion of the theorem is true by the choice of $c$ and Lemma~\ref{lem-basis}.
So we prove it by induction on $n$ for $n>s^3 \ge 256$. Notice, $n\ge M \ge B^2$ so $\sqrt{n} \le n/B$.

By recurrence~(\ref{rec-IO}):
\begin{eqnarray}
    T(n) \le m \cdot T(m) + \E_{(n_1,\dots,n_m)\sim \mu_n^m}\left[ \sum_{i=1}^m T(n_i) \right] + c_L\left(1+\frac{n}{B}\right).
\end{eqnarray}

First we bound $m \cdot T(m)$. We use the following simple observations:
$m\le 2\sqrt{n}$, $m^2 \le n + 2\sqrt{n}+1$, $\log_s m \le \log_s \sqrt{n} + \log_s 2$, and $\sqrt{n}\log \sqrt{n} \le n/4$.  
By the induction hypothesis:
\begin{eqnarray*}
m\cdot T(m) &\le&  m\cdot \frac{cm}{B} \left(2 \max(\log_{s} m,1) - 1\right) + m\cdot \frac{c}{16} \\
            &\le&  \frac{2cm^2}{B} \cdot \log_s m - \frac{cn}{B} + m\cdot \frac{c}{16} \\
            &\le&  \frac{2c n}{B} \cdot \log_s \sqrt{n} + \frac{2cn}{B} \cdot \log_s 2 \\
            &&+    \frac{4c \sqrt{n}}{B} \cdot \log_s \sqrt{n} + \frac{4c\sqrt{n}}{B} \cdot \log_s 2 \\
            &&+    \log_s \sqrt{n} + \log_s 2 + \frac{c\sqrt{n}}{8}  - \frac{cn}{B} \\
            &\le&  \frac{c n}{B} \cdot \log_s n + \frac{5cn}{B \log s} + \frac{cn}{8B} - \frac{cn}{B}.
\end{eqnarray*}

Now we bound $\E_{(n_1,\dots,n_m)\sim \mu_n^m}\left[ \sum_{i=1}^m T(n_i) \right]$ using induction and Corollary~\ref{c-expectation}:
\begin{eqnarray*}
\E_{(n_1,\dots,n_m)\sim \mu_n^m}\left[ \sum_{i=1}^m T(n_i) \right] 
            &\le&  2\cdot \frac{c}{B} \E_{(n_1,\dots,n_m)\sim \mu_n^m}\left[ \sum_{i=1}^m n_i\cdot \max(\log_{s} n_i,1)\right] - \frac{cn}{B} + \frac{cm}{16}\\
            &\le&  \frac{c}{B} n \log_s n + \frac{2c}{B} \cdot \frac{4\e n}{\log s} + \frac{2c}{B} \cdot 2\e^2s^2  - \frac{cn}{B} + \frac{cn}{8B}.
\end{eqnarray*}

Using $\log s \ge \log B \ge \log \beta \ge 8(5+8\e+4\e^2)$, we get as an upper bound on $T(n)$:
$$
T(n) \le 2\frac{c}{B} n \log_s n + \frac{3}{8}\cdot \frac{cn}{B} - 2 \frac{cn}{B} + c_L\left(\frac{n}{B}+1\right) \le \frac{2c}{B} n \log_s n  -  \frac{cn}{B}
$$
since $n/B>1$ and $\frac{3c}{8} + 2 c_L \le \frac{c}{2}.$ This proves the theorem. 
\end{proof}

\section{Experiments}
To provide a comparison among SquareSort and other sorting algorithms, we compare SquareSort with std::sort and FunnelSort. The first algorithm is a part of the C++ standard library defined in header <algorithm> on g++ and implemented as an IntroSort algorithm. The IntroSort algorithm is a hybrid sort algorithm that combines QuickSort and HeapSort. FunnelSort is another cache-oblivious algorithm; we use its implementation by Frederik Rønn~\cite{ronn} which is also written in C++.

We will compare the time each algorithm takes to sort an array of integers. In each step, we want to sort arrays of the total size of one-third of the memory. The arrays will consist of 32-bit signed integers. Since both std::sort and SquareSort are Las Vegas algorithms, the running time is a random variable. We repeat each test on multiple instances and take the average running time. All tests are run on the Linux operating system, the algorithms are written in C++ and compiled by the g++ compiler. 

We start with the size of $1000$ elements and in each round, we proportionally increase the size of the arrays. We will compare totally four distinct distributions of input elements: a random permutation of numbers in $\{1, \dots, n\}$, a random sequence of binary values, a random sequence of integers from $\{1,\dots, n\}$ selected uniformly at random, and a sequence of integers selected uniformly at random from the range $\{1, \dots, \sqrt{n}\}$. 
We tested the algorithms on an AMD Ryzen 7 1800X Eight-Core Processor 
with three levels of caches with sizes of 96K (L1 per core), 512K (L2 per core) and 16MB (L3 shared) respectively and 32 GB of main memory.
(Measurements on other systems gave similar looking results.)
In the implementation of SquareSort whenever the size of an array is less than 1000 elements we sort it directly using std::sort, also in procedure $\SkewTranspose$, we transpose elements directly whenever that given region has less than 10 columns or the number of buckets is less than 10.

\subsection{Results}
For each size, we measure the average time in nanoseconds. As all three algorithms have the same asymptotic time complexity, we normalize the measured average time $t$ as $t / n \log n$, where $n$ is the size of the sorted array. We plot this normalized time per item as it depends on the number of elements $n$.

For each type of array, std::sort was the fastest, then SquareSort, and last came the FunnelSort. As in the SquareSort, we split the problem into approximately $\sqrt{n}$ problems of size $\sqrt{n}$, this is the reason why we can observe a sudden increase around $10^6$, since here we add another recursive call in expectation. 

        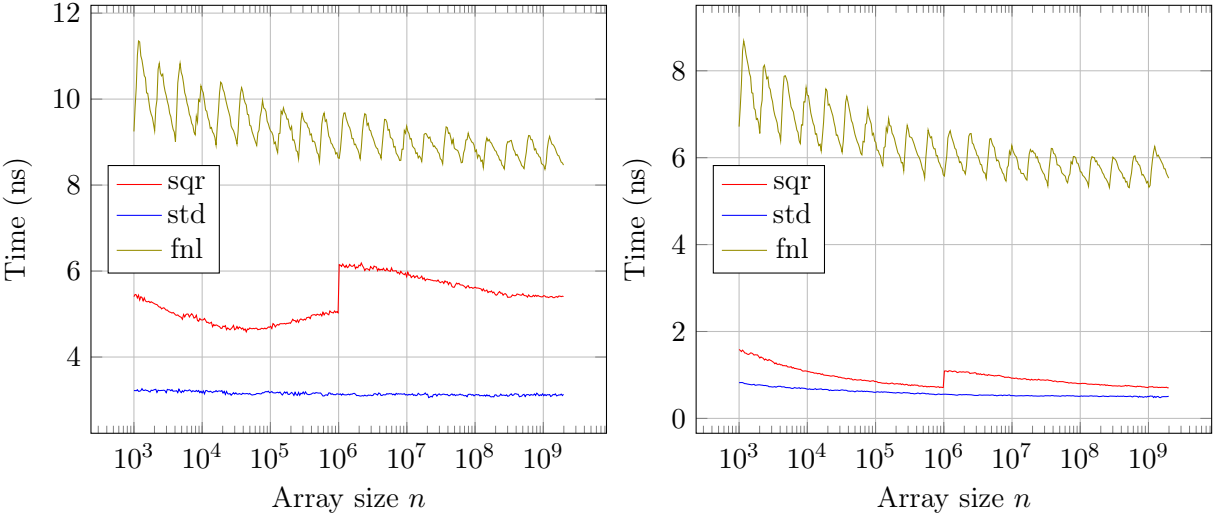
\begin{figure}[H]
            \begin{tikzpicture}
                \begin{axis}[
                    xmode = log,
				    xlabel = Array size $n$,
				    ylabel = Time (ns),
                    grid=major,
                    legend style={at={(0.25,0.5)},anchor=east},
				    ]
				    \addplot [red, mark=none] table [x=size, y=div, col sep=space] 
                    {mereni-13-07-mereni-finalsqare_sort-perm-div.csv};
				    \addplot [blue,  mark=none] table [x=size, y=div, col sep=space] {mereni-13-07-mereni-finalsort-perm-div.csv};
                    \addplot [olive,  mark=none] table [x=size, y=div, col sep=space] {mereni-13-07-mereni-finalfunnel_sort-perm-div.csv};
                    \addlegendentry{sqr};
                    \addlegendentry{std};
                    \addlegendentry{fnl};
                \end{axis}
            \end{tikzpicture}
            \begin{tikzpicture}
                \begin{axis}[
                    xmode = log,
				    xlabel = Array size $n$,
				    ylabel = Time (ns),
                    grid=major,
                    legend style={at={(0.25,0.5)},anchor=east},
				    ]
				    \addplot [red, mark=none] table [x=size, y=div, col sep=space] 
                    {mereni-13-07-mereni-finalsqare_sort-binary-div.csv};
				    \addplot [blue,  mark=none] table [x=size, y=div, col sep=space] {mereni-13-07-mereni-finalsort-binary-div.csv};
                    \addplot [olive,  mark=none] table [x=size, y=div, col sep=space] {mereni-13-07-mereni-finalfunnel_sort-binary-div.csv};
                    \addlegendentry{sqr};
                    \addlegendentry{std};
                    \addlegendentry{fnl};
                \end{axis}
            \end{tikzpicture}
            \caption{Time per item to sort a random permutation (left) and a random binary sequence (right).}
        \end{figure}

        \begin{figure}[H]
            \begin{tikzpicture}
                \begin{axis}[
                    xmode = log,
				    xlabel =  Array size $n$,
				    ylabel = Time (ns),
                    grid=major,
                    legend style={at={(0.25,0.5)},anchor=east},
				    ]
				    \addplot [red, mark=none] table [x=size, y=div, col sep=space] 
                    {mereni-13-07-mereni-finalsqare_sort-sqrt-div.csv};
				    \addplot [blue,  mark=none] table [x=size, y=div, col sep=space] {mereni-13-07-mereni-finalsort-sqrt-div.csv};
                    \addplot [olive,  mark=none] table [x=size, y=div, col sep=space] {mereni-13-07-mereni-finalfunnel_sort-sqrt-div.csv};
                    \addlegendentry{sqr};
                    \addlegendentry{std};
                    \addlegendentry{fnl};
                \end{axis}
            \end{tikzpicture}
            \begin{tikzpicture}
                \begin{axis}[
                    xmode = log,
				    xlabel =  Array size $n$,
				    ylabel = Time (ns),
                    grid=major,
                    legend style={at={(0.25,0.5)},anchor=east},
				    ]
				    \addplot [red, mark=none] table [x=size, y=div, col sep=space] 
                    {mereni-13-07-mereni-finalsqare_sort-full-div.csv};
				    \addplot [blue,  mark=none] table [x=size, y=div, col sep=space] {mereni-13-07-mereni-finalsort-full-div.csv};
                    \addplot [olive,  mark=none] table [x=size, y=div, col sep=space] {mereni-13-07-mereni-finalfunnel_sort-full-div.csv};
                    \addlegendentry{sqr};
                    \addlegendentry{std};
                    \addlegendentry{fnl};
                \end{axis}
            \end{tikzpicture}
            \caption{Time per item to sort a random sequence of elements from the universe of size $n$ (left) and of size $\sqrt{n}$ (right).}
        \end{figure}
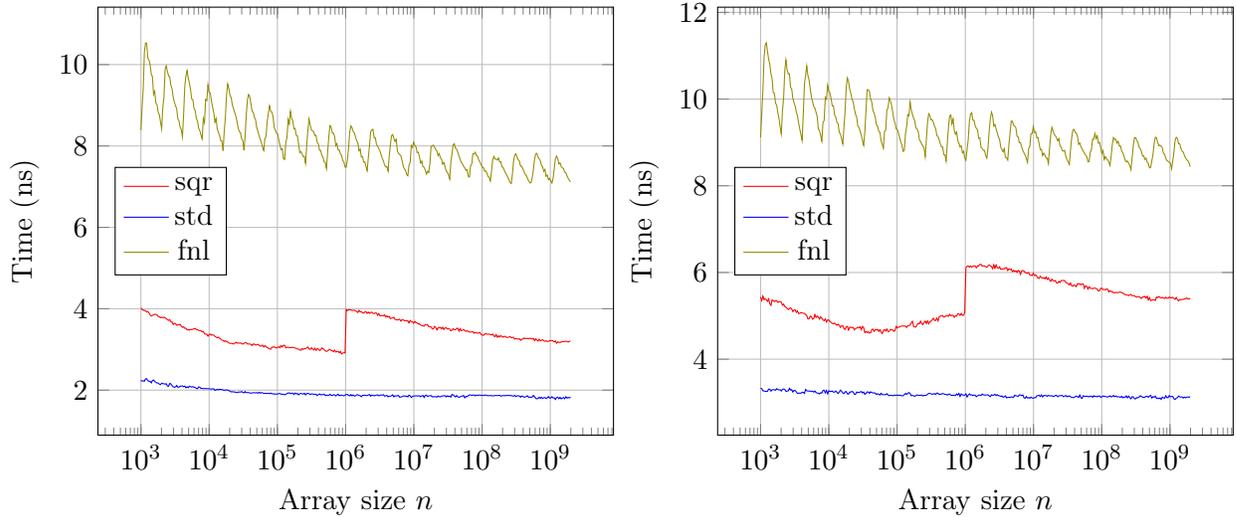

\subsection{Cutoff}
One of the parameters in the square sort algorithm is the size of an array that we sort directly by std::sort at the bottom of the recursion; we will call this parameter {\em cutoff}. We have tested the previous experiment on multiple different cutoffs ranging from 100 to 958. We were interested in how this parameter affects the running time. We present one graph with four cutoffs: 100, 256, 493, and 958. Again we normalize the running time for each size.

The cutoff parameter mainly determines at what size we add an additional recursive call to square sort. At cutoff 100 the additional call happens around $10^4$ items and then next at $10^8$ items. As we increase the cutoff the additional call is added later and for 958 the call is added around one million.

        \begin{figure}[H]
            \begin{tikzpicture}
                \begin{axis}[
                    xmode = log,
				    xlabel =  Array size $n$,
				    ylabel = Time (ns),
                    grid=major,
                    legend style={at={(0.24,0.8)},anchor=east},
				    ]
                    \addplot [olive, mark=none] table [x=size, y=div, col sep=space] 
                    {cutoff-tables-cutoff=100.csv};
                    \addplot [blue, mark=none] table [x=size, y=div, col sep=space] 
                    {cutoff-tables-cutoff=256.csv};
                    \addplot [red, mark=none] table [x=size, y=div, col sep=space] 
                    {cutoff-tables-cutoff=493.csv};
                    \addplot [teal, mark=none] table [x=size, y=div, col sep=space] 
                    {cutoff-tables-cutoff=958.csv};
                    \addlegendentry{100};
                    \addlegendentry{256};
                    \addlegendentry{493};
                    \addlegendentry{958};
                \end{axis}
            \end{tikzpicture}
            \caption{Time per item to sort a random permutation with different cutoffs.}
        \end{figure}
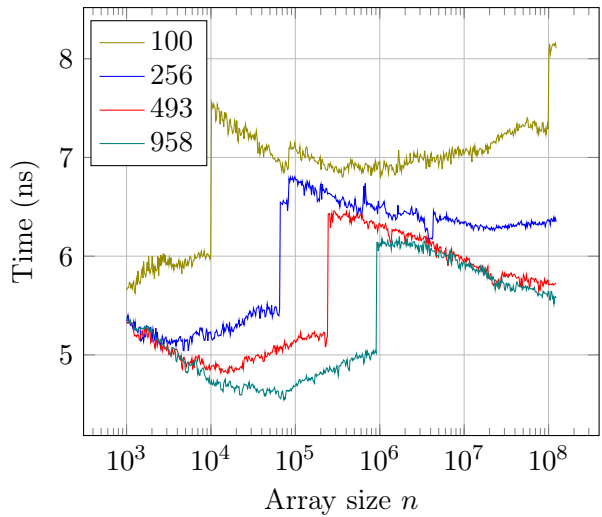

\subsection{External sorting}
In order to test our algorithm in more diverse hierarchy of memories we compare the three algorithms when sorting data stored on a disk.
In this experiment we created a memory mapped file of size $s$ where for $n<5\cdot 10^7$, $s=4GB$, for $5\cdot 10^7 \le n < 4 \cdot 10^9$, $s=64GB$ and
for $n\ge 4 \cdot 10^9$, $s=128GB$.
For each $n$ from $1000$ till $16\cdot 10^9$ we filled the whole file with random $64$-bit integers, subdivided it into blocks of size $n$, and sorted each of the blocks using one of the three algorithms. We took the average running time over the blocks.

The tests were performed on a computer running Linux version 5.10.0-29-amd64 equipped by AMD Ryzen 5 7600 6-Core Processor with 8 GB of RAM and an ssd disk Samsung SSD 970 EVO Plus 1TB using ext4 filesystem.
We plot the resulting measurements next.
The graphs are again normalized, the time represents time spent on sorting one element of the array, the $x$ axis corresponds to array size $n$.

        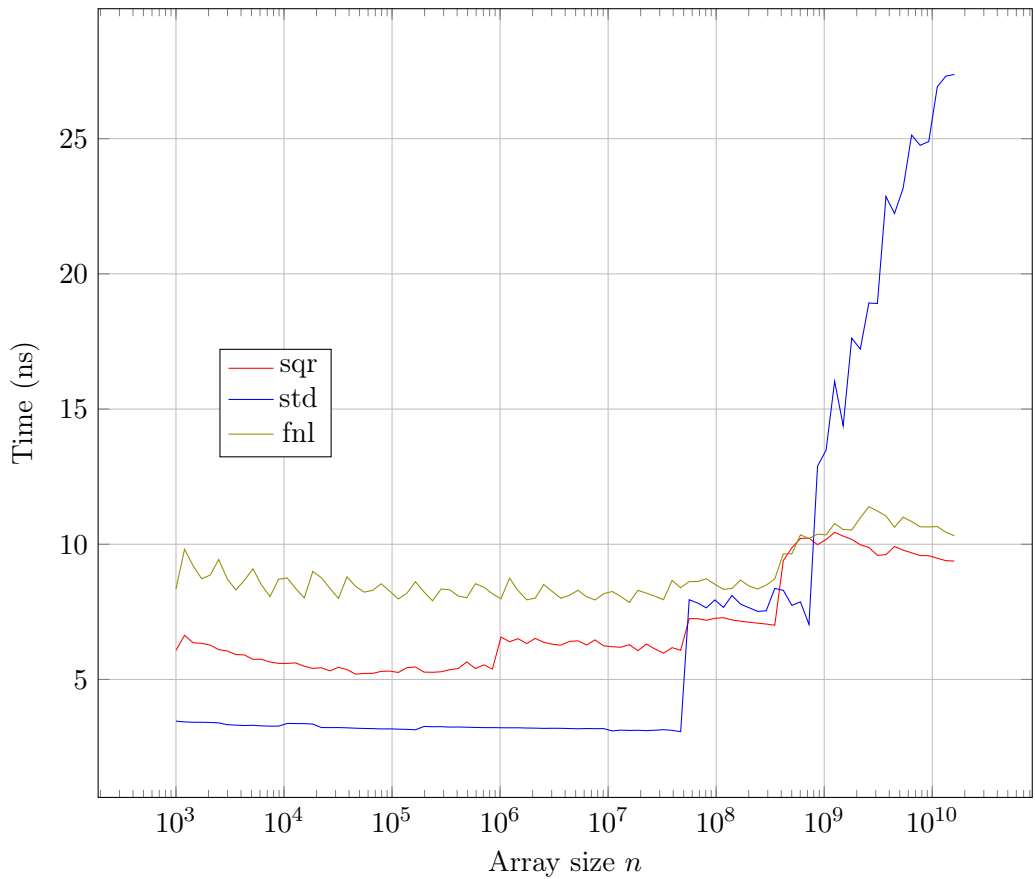
\begin{figure}[H]
            \begin{tikzpicture}
                \begin{axis}[
                    width=14cm,
                    xmode = log,
				    xlabel =  Array size $n$,
				    ylabel = Time (ns),
                    grid=major,
                    legend style={at={(0.25,0.5)},anchor=east},
				    ]
				    \addplot [red, mark=none] table [x=size, y=norm, col sep=space] 
                    {zakami_processed-sq_full_zakami_plus.csv};
				    \addplot [blue,  mark=none] table [x=size, y=norm, col sep=space] {zakami_processed-std_full_zakami_plus.csv};
                    \addplot [olive,  mark=none] table [x=size, y=norm, col sep=space] {zakami_processed-fnl_full_zakami_plus.csv};
                    \addlegendentry{sqr};
                    \addlegendentry{std};
                    \addlegendentry{fnl};
                \end{axis}
            \end{tikzpicture}
            \caption{Comparison of all three algorithms in external sorting experiment.}
        \end{figure}

For $n<5 \cdot 10^7$ the whole mapped file fits in the memory and the relative speed of the three algorithms agrees with our previous measurements: The std::sort is fastest, followed by SquareSort and then by FunnelSort. 
We can see an increase in time taken by the algorithms when the whole file does not fit into memory around size of $n=5 \cdot 10^7$.
Another increase in the time occurs when the sorted array does not fit into memory.
Since neither SquareSort nor FunnelSort are in-place sorting algorithms this happens for them sooner than for std::sort.
However, once std::sort does not fit the sorted array in its memory we can observe sharp increase in its cost starting from about $n=10^9$ elements.

\section*{Acknowledgements}

The measurements were performed on desktop computers at the Computer Science Institute of Charles University. We thank our technical staff for their support with our measurements. 

\section{Appendix}
\lstset{language=C++,
basicstyle=\ttfamily,
keywordstyle=\color{blue}\ttfamily,
stringstyle=\color{red}\ttfamily,
commentstyle=\color{gray}\ttfamily,
morecomment=[l][\color{magenta}]{\#}}
\begin{lstlisting}[language=C++,
caption={Square sort implementation},
% label=square\_sort.h,
columns=fullflexible,
frame=single,
breaklines=true,
postbreak=\mbox{\textcolor{black}{$\hookrightarrow$}\space}]
#include "stdio.h"
#include <time.h>
#include <bits/stdc++.h>
using namespace std;

typedef int T;

void add_bucket_sizes(T A[], int size, T pivots[], int pivot_cnt, T bucket_size[])
// pivots[pivot_cnt-1] = INFTY;
{
   int i=0;
   int j=0;

   pivot_cnt--;

   for(;i<pivot_cnt;i++){
     while((j < size) && (A[j]<=pivots[i])){ bucket_size[i]++; j++; }
   }
   bucket_size[pivot_cnt]+=size-j;
}

void sample_pivots(T A[], int size, T sample[], int sample_cnt)
{
   int i;

   for(i=0; i<sample_cnt; i++)sample[i]=A[((unsigned int) rand())%size];

   sort(sample, sample+sample_cnt);
   sample[sample_cnt-1]=INT_MAX;

   for(i=0; i<sample_cnt-1; i++)if(sample[i]==sample[i+1])sample[i]--;

}

void naive_skew_transpoze(T A[], T B[], int col_start[], int col_end[], int col_cnt, int bucket_start[], int bucket_cnt, T pivots[])
{
    int i=0;
    int j=0;
     
    for(i=0; i<bucket_cnt; i++){
       for(j=0; j<col_cnt; j++){
           while(( col_start[j]  < col_end[j] ) && (A[ col_start[j] ] <= pivots[i])){
               B[ bucket_start[i] ] = A[ col_start[j] ];
               bucket_start[i]++;
               col_start[j]++;
            }
        }
    }
}


void skew_transpoze(T A[], T B[], int col_start[], int col_end[], int col_cnt, int bucket_start[], int bucket_cnt, T pivots[])
{
   if((col_cnt < 10)||(bucket_cnt < 10)){ 
      naive_skew_transpoze(A, B, col_start, col_end, col_cnt, bucket_start, bucket_cnt, pivots); 
      return;
   }

   int half_col_cnt = col_cnt /2;
   int half_bucket_cnt = bucket_cnt /2;

   skew_transpoze(A, B, col_start, col_end, half_col_cnt, bucket_start, half_bucket_cnt, pivots);
   skew_transpoze(A, B, col_start+half_col_cnt, col_end+half_col_cnt, col_cnt - half_col_cnt, bucket_start, half_bucket_cnt, pivots);
   skew_transpoze(A, B, col_start, col_end, half_col_cnt, bucket_start + half_bucket_cnt , bucket_cnt - half_bucket_cnt, pivots + half_bucket_cnt );
   skew_transpoze(A, B, col_start+half_col_cnt, col_end+half_col_cnt, col_cnt - half_col_cnt, bucket_start + half_bucket_cnt, bucket_cnt - half_bucket_cnt, pivots + half_bucket_cnt );

}

void square_sort(T A[], T B[], int size, int col_start[], int col_end[], int bucket_start[], T pivots[], int buf_size)
// col_start[], col_end[], bucket_start[], pivots[] are buffers of buf_size items.
{

  int i,j;

   if(size < 1000){ 
      for(i=0;i<size;i++)B[i]=A[i];
      sort(B,B+size);
      return;
   }   

   int new_size = sqrt(size);

   sample_pivots(A, size, pivots, new_size);
  
   for(i=0;i<new_size;i++){ col_start[i]=new_size*i; col_end[i]=new_size * (i+1); bucket_start[i]=0; }
   bucket_start[new_size]=bucket_start[new_size+1]=0;

   col_end[new_size-1] = size;

   if(buf_size <= new_size+2)throw std::logic_error("Not enough memory for recursion");

   for(i=0;i<new_size;i++){
      square_sort(A+col_start[i],B+col_start[i], col_end[i] - col_start[i], col_start+new_size, col_end+new_size, bucket_start+new_size+2, pivots+new_size, buf_size - new_size-2);
      add_bucket_sizes( B+col_start[i], col_end[i] - col_start[i], pivots, new_size, bucket_start+2);
   }

   for(i=2;i<new_size+1;i++){ bucket_start[i] += bucket_start[i-1]; }

   skew_transpoze(B, A, col_start, col_end, new_size, bucket_start+1, new_size, pivots);

   for(i=0;i<new_size;i++){
      if((i>0)&&(pivots[i]==pivots[i-1]+1)){  // monochromatic bucket - just copy
           for(j=bucket_start[i];j<bucket_start[i+1];j++)B[j]=A[j];         
      }
      else square_sort(A+bucket_start[i],B+bucket_start[i], bucket_start[i+1] - bucket_start[i], col_start+new_size, col_end+new_size, bucket_start+new_size+2, pivots+new_size, buf_size - new_size-2);
   }

}


\end{lstlisting}
\nocite{*}

\bibliographystyle{alpha}
\bibliography{references}

\end{document}